\documentclass[12pt]{article}
\usepackage[dvipdfmx,hiresbb]{graphicx}
\usepackage{xcolor}
\usepackage{amsmath}
\usepackage{amssymb}
\usepackage{mathrsfs}
\usepackage{amsthm}
\usepackage{amscd}
\usepackage{txfonts}      
\usepackage{bbm}
\usepackage{dsfont}
\usepackage{anyfontsize}

\usepackage{sectsty}

\sectionfont{\Large}
\subsectionfont{\large}
\subsubsectionfont{\normalsize}

\newcommand{\dia}{\star}
\newcommand{\e}{\epsilon}

\newcommand{\G}{\mathbf{G}}
\renewcommand{\P}{\bo{P}}

\newcommand{\C}{\mathbb{C}}
\renewcommand{\H}{\mathbb{H}}
\newcommand{\lap}{\text{\Large $\square$}}

\newcommand{\va}{\varphi}

\newcommand{\w}{\wedge}

\newcommand{\vs}{\vspace}
\newcommand{\Tr}{\operatorname{Tr}}

\newcommand{\Ker}{\operatorname{Ker}}

\renewcommand{\Im}{\operatorname{Im}}

\newcommand{\vol}{\operatorname{vol}}

\newcommand{\bo}[1]{\boldsymbol{#1}}
\renewcommand{\b}[1]{\bar{#1}}

\newcommand{\pa}{\partial}

\newcommand{\A}{\mathcal{A}}
\newcommand{\B}{\mathcal{B}}

\newcommand{\disp}{\displaystyle}
\newcommand{\tsum}{\textstyle{\sum}\disp}
\newcommand{\id}{\operatorname{id}}
\renewcommand{\d}{\operatorname{d}}
\renewcommand{\i}{\operatorname{i}}
\newcommand{\M}{\mathcal{M}}

\theoremstyle{definition}
\newtheorem{lemma}{Lemma}
\newtheorem{proposition}{Proposition}
\numberwithin{equation}{section}

\begin{document}




\begin{center}
\textbf{\LARGE
 Algebra of Kodaira-Spencer Gravity}
 \vs{0.2cm}

 \textbf{\LARGE
 and Deformation of Calabi-Yau Manifold}
 \vs{1cm}
 
 Kenji Mohri\\
Faculty of Pure and Applied Sciences,
University of Tsukuba\\ 1-1-1 Ten-nodai,
Tsukuba, Ibaraki 305-8571, Japan\\
\emph{mohri@het.ph.tsukuba.ac.jp}

\end{center}


\begin{abstract}
 We study the algebraic structure of the configuration space
 of the Kodaira-Spencer gravity theory on a Calabi-Yau threefold.
 We then investigate the deformation problem of the Kodaira-Spencer gravity
 at the classical level using the algebraic tools obtained here. 
\end{abstract}



\section{Introduction}	
Generally speaking, second quantization of string theory is string field theory \cite{CBTh,Zw}.
In the case of topological closed B string theory \cite{Wi1}, however,
 its second quantization is known to
reduce to a six-dimensional field theory on a Calabi-Yau three-fold $X$ \cite{Wi2},
which is the Kodaira-Spencer (KS) theory of gravity \cite{BCOV}.
The configuration space of the KS gravity theory $\B$, which we will define in \eqref{b} below,
 has a structure of the differential Gerstenhaber-Batalin-Vilkovisky (dGBV in short) algebra \cite{BaKo,Man}.

 In this paper we study the algebraic structure of the configuration space $\B$ more closely.
 We show, in particular,  that some operators such as  the Hodge dual operator or the Lefshetz operators,
which are originally defined on the space of the differential forms $\A$ and are transferred to $\B$
 by the isomorphism between $\A$ and $\B$ as vector spaces, 
 behave better in $\B$ than in $\A$ in some sense.
 We also find another dGBV algebra structure in $\B$.
 Then we consider the problem of holomorphic deformation of the KS theory
 using the algebraic tools developed above.

 The organization of the present paper is as follows.
 In Section 2, the configuration space $\B$ and the action of the KS gravity theory,
 as well as the space of the differential forms $\A$, are introduced.
  In Section 3, we describe the algebraic structure of the configuration space $\B$;
  we convert the various linear operators on $\A$ to $\B$ and express these as differential operators
  of the sigma model variables;
  in particular we
  find that the Hodge dual operator is super-algebra homomorphism modulo a phase factor,
and the Lefshetz operators are derivations on $\B$.
We give another dGBV algebra structure  on $\B$;
two dGBV algebras are related to each other by the Hodge duality.
In Section 4, we construct the solution of the KS equation, which is the classical equation of the KS action,
using the algebraic tools developed above.
  Then we describe the deformation of the action and the supercharges
  associated with the condensation of the solution.
In Section 5, we identify the deformation of the previous section as the holomorphic limit
   of the deformation of the complex structure of the Calabi-Yau three-fold $X$
   by a direct computation.
   We also propose a deformation of the states under the holomorphic deformation.
In Section 6, we discuss rather briefly future directions of the present  paper.

\section{Kodaira-Spencer Theory of Gravity}
\subsection{Calabi-Yau threefold}
Let $X$ be a Calabi-Yau threefold with a fixed complex structure.
We also fix a holomorphic three-form $\varOmega=s(z)\d z^1\w\d z^2\w\d z^3$,
where $z^i$s are local holomorphic coordinates of $X$, and  
a K\"{a}hler form $\omega=\i \sum g_{i\b{j}}\d z^i\w\d z^{\b{j}}$.
We do not assume that the K\"{a}hler metric $g_{i\b{j}}$ is Ricci-flat.

There are two volume forms:
the one is $\d \vol_{\omega}=\omega^3/3!$, the other is
$\d \vol_{\varOmega}=\i \varOmega\w\b{\varOmega}$,
the ratio of which yields a scalar function  $\sigma$ on $X$:
\begin{equation}
 \exp(\sigma)=|s|^2(\det g)^{-1}.
\label{scalar}
\end{equation}
Note that the Ricci tensor can be written as  $R_{i\b{j}}=\pa_i\pa_{\b{j}}\sigma$.

\subsection{Configuration space of the Kodaira-Spencer theory}
Two dimensional B-twisted topological sigma model with $X$ as a target \cite{Wi1} coupled to topological gravity
defines a topological B-string theory \cite{Wi2,BCOV}.

The string field theory of the topological B-string is reduced to a six dimensional field theory on $X$ \cite{Wi2},
known as the Kodaira-Spencer  theory of gravity \cite{BCOV}.

To explain the field content of the KS theory, we introduce the following space:
\begin{equation}
 \B=\bigoplus_{p,q=0}^{3} B^{p,q},
  \quad B^{p,q}=\Gamma(X,\textstyle{\bigwedge}^p TX\otimes \textstyle{\bigwedge}^q\overline{T^{\ast}X}),
  \label{b}
\end{equation}
where $TX(\overline{T^{\ast}X})$ is the (anti-)holomorphic (co)tangent bundle of $X$.
Note that $\B$ is a subspace of the state space of the topological B-sigma model \cite{Wi1}
 which has $X$ as the target with all stringy excitations supressed \cite{Wi2}.

%
It is convenient to use the sigma model variables $\bo{\theta}_i=\pa/\pa z^i $, $\bo{\eta}^{\b{j}}=\d z^{\b{j}}$
so that an element $\alpha$ of $B^{p,q}$ is expressed as
\begin{equation}
 \alpha=\frac{1}{p!q!}
\sum
  \alpha^{i_1\cdots i_p}_{\b{j}_1\cdots\b{j}_q}(z,\b{z})\,
  \bo{\theta}_{i_1}\cdots \bo{\theta}_{i_p}
  \bo{\eta}^{\b{j}_1}\cdots \bo{\eta}^{\b{j}_q}.
\end{equation}
For $\alpha\in B^{p,q}$, $|\alpha|=p+q$ is called the ghost number of $\alpha$.



 \subsection{Space of the differential forms}
We have  also the space of differential forms on $X$:
\begin{equation}
 \A=\bigoplus_{p,q=0}^{3} A^{p,q}, \quad
  A^{p,q}=\Gamma(X,\textstyle{\bigwedge}^p T^{\ast}X\otimes \textstyle{\bigwedge}^q\overline{T^{\ast}X}).
\end{equation}
We define the Hermite metric on $\A$  by the Hodge dual operator \ $\ast:A^{p,q}\to A^{3-q,3-p}$ as  
\begin{equation}
 (a|b)=\int_Xa\w \ast\b{b}, \quad a,b\in \A.
  \label{hermite}
\end{equation}
The exterior differential operators are given by
$\pa=\sum\d z^i \w \pa/\pa z^i$,
$\b{\pa}=\sum\d z^{\b{i}} \w \pa/\pa z^{\b{i}}$,
and their conjugates with respect to the metric (\ref{hermite}) are
$\pa^{\dag}=-\ast \b{\pa} \ast$,
$\b{\pa}^{\dag}=-\ast \pa \ast$.

These four operators are all nilpotent and their anti-commutators all vanish
except for two pairs,
 both of which give the  same Laplacian:
\begin{equation}
 \lap=\b{\pa}\b{\pa}^{\dag}+\b{\pa}^{\dag}\b{\pa}=\pa\pa^{\dag}+\pa^{\dag}\pa.
\end{equation}
The harmonic states are its kernel: $\mathbb{H}=\Ker \lap$,
the orthogonal complement of which with respect to the metric (\ref{hermite})
admits further orthogonal decompositions:
\begin{equation}
 \mathbb{H}^{\perp}=\Im\b{\pa}\oplus \Im\b{\pa}^{\dag}=\Im\pa\oplus \Im\pa^{\dag}.
 \label{hodgedecomposition}
\end{equation}
In particular  the subspace $\Ker\pa\subset \A$ 
is decomposed as
\begin{equation}
 \Ker\pa=\mathbb{H}\oplus \Im(\pa\b{\pa})\oplus \Im(\pa\b{\pa}^{\dag}),
\end{equation}
where $\mathbb{H}$ is called the physical states, $\Im(\pa\b{\pa})$ the trivial states, 
and $\Im(\pa\b{\pa}^{\dag})$ the unphysical states in \cite{Zw}
in the context of closed string field theory.

The Green's operator $\G$ is defined to be zero on $\mathbb{H}$, and $\lap^{-1}$ on $\mathbb{H}^{\perp}$ \cite{Kod}.
Let $\Pi:\A\to\H$ be the orthogonal projection onto the harmonic forms, then
\begin{equation}
 \id_{\A} =\Pi+\G\lap.
\label{idA}
\end{equation} 

Form the Lefshetz operators 
$L=\omega\w :A^{p,q}\to A^{p+1,q+1}$,  $\varLambda=L^{\dag}=\ast^{-1}\circ L\circ \ast$,
we can form  a $S\!U(2)$ algebra which is self-adjoint with respect to the metric (\ref{hermite}):
\begin{equation}
 J_1=\frac12(L+\varLambda), \quad
  J_2=\frac{1}{2\i}(\varLambda-L),\quad
  J_3=\frac12(3-P-Q),
  \label{angmom}
\end{equation}
where $(P,Q)|_{A^{p,q}}=(p,q)$, and $[J_a,J_b]=\i\sum \e_{abc}J_c$. 

The commutation relations among the differential operators and the Lefshetz operators above are known as
the Hodge-K\"{a}hler identities \cite{Kob}:
\begin{alignat}{2}
 [L,\pa]&=0, &\quad [L,\b{\pa}]&=0,
 \label{hodgekahlerone}\\
 [\varLambda,\pa^{\dag}]&=0, &\quad [\varLambda,\b{\pa}^{\dag}]&=0,\\
 [L,\pa^{\dag}]&=\i\b{\pa}, &\quad [L,\b{\pa}^{\dag}]&=-\i \pa,\\
 [\varLambda,\pa]&=\i\b{\pa}^{\dag}, &\quad [\varLambda,\b{\pa}]&=-\i \pa^{\dag}.
 \label{hodgekahlerfour}
\end{alignat}

Using the holomorphic three-form $\varOmega$, we can define 
 an isomorphism $\rho:\B\to \A$ of vector spaces \cite{Hu}
 \begin{align}
  \rho(1)&=s\d z^1\w\d z^2\w \d z^3, \label{isomone} \\
  \rho(\bo{\theta}_i)&=\tfrac12 s
  \tsum_{jk}
  \e_{i\!jk}\d z^j\w\d z^k,  \\
  \rho(\bo{\theta}_i\bo{\theta}_j)&=-s\tsum_k
  \e_{i\!jk}\d z^k, \\
  \rho(\bo{\theta}_1\bo{\theta}_2\bo{\theta}_3)&=-s,
  \label{isomfour}
 \end{align}
 where $\e_{ijk}$ is the Levi-Civita symbol.
Note that $\rho$ maps $B^{p,q}$ to $A^{3-p,q}$, so that it inverts the Grassmann parity. 

The significance of the isomorphism $\rho$ is that
we can convert each linear operator $f$ on $\A$ to
 that of $\B$ by $\hat{f}:=\rho^{-1}\circ f\circ \rho$.
 In particular, we give the special names for the operators below
  \begin{alignat}{2}
  \Delta&=\hat{\pa}& :   B^{p,q}&\to B^{p-1,q}, \label{delta}\\
 S&=\widehat{\pa^{\dag}} &:   B^{p,q}&\to B^{p+1,q}, \label{es}\\
 R&=\widehat{\b{\pa}^{\dag}}  &:   B^{p,q}&\to B^{p,q-1}. \label{ar}
  \end{alignat}
  In fact, $\B$ itself has a Hermite metric, $\b{\pa}_{\B}$, which is the BRST operator of the topological string \cite{Wi1},
  and its conjugate $\b{\pa}^{\dag}_{\B}$ \cite{Kod,Kob}.
  Note that here  we are  temporarily using subscripts in order to distinguish $\b{\pa}_{\B}$ from $\b{\pa}_{\A}$.
  Then we  find that
  \begin{equation}
  \b{\pa}_{\B}=-\widehat{\b{\pa}_{\A}}, \quad \b{\pa}^{\dag}_{\B}=-e^{\sigma} \circ R\circ e^{-\sigma}, 
  \end{equation}
  where $\sigma$ is the scalar function defined in (\ref{scalar}) \cite{Hu}.

  \subsection{Kodaira-Spencer gravity action}
  Let us  define the trace map  $\Tr:\B\to \C$ (complex numbers) by
  \begin{equation}
   \Tr(\alpha)=\int_X\rho(\alpha)\w \varOmega, \quad \alpha\in \B.
  \end{equation}
  It is clear that $\Tr(\alpha)\ne 0$ only if $\alpha\in B^{3,3}$.

\begin{lemma}
 Let $\alpha\in B^{p,q}$, $\beta\in B^{3-p,3-q}$, \ then \
 $\rho(\alpha)\w\rho(\beta)=(-1)^{q+1}\rho(\alpha\w\beta)\w\varOmega$.
\end{lemma}
\begin{proof}
 We use ordered multi-indices to write the elements
 \begin{equation*}
  \alpha=\sum_{I,J}^{<}\alpha^{I}_{\b{J}} \bo{\theta}_{I}\bo{\eta}^{\b{J}}\in B^{p,q}, \quad
\beta=\sum_{K,L}^{<}\beta^{K}_{\b{L}}\bo{\theta}_{K}\bo{\eta}^{\b{L}}\in B^{3-p,3-q},
 \end{equation*}
 where each multi-index has the form $I=(i_1,\dots,i_p)$,
 $i_1<i_2<\cdots<i_p$.
 For $I=(i_1,\dots,i_p)$, let $I^{\ast}=(i^{\ast}_1,\dots, i^{\ast}_{3-p})$ be the complementary multi-index, i.e.,
$\{i_1,\cdots,i_p,i^{\ast}_1,\dots, i^{\ast}_{3-p}\}=\{1,2,3\}$ as an unordered set.
 Then we have
 $\rho(\bo{\theta}_I)=(-1)^{|I|(|I|-1)/2} \e_{I,I^{\ast}} s\bo{\eta}^{I^{\ast}}$, where $\bo{\eta}^i=\d z^i$,

Simple computation shows that
   \begin{align*}
  \rho(\alpha)\w\rho(\beta)
  &=(-1)^{pq}s
   \sum_{I,J}^{<} \e_{I,I^{\ast}}\e_{J,J^{\ast}} \alpha^{I}_{\b{J}} \,\beta^{I^{\ast}}_{\b{J}^{\ast}}
   \bo{\eta}^{\b{1}}\bo{\eta}^{\b{2}}\bo{\eta}^{\b{3}}\w\varOmega,\\
  \alpha\w\beta
&  =(-1)^{q+pq}\sum_{I,J}^{<}
  \e_{I,I^{\ast}}\e_{J,J^{\ast}}
   \alpha^I_{\b{J}}\,\beta^{I^{\ast}}_{\b{J}^{\ast}}
   \bo{\theta}_1\bo{\theta}_2\bo{\theta}_3
  \bo{\eta}^{\b{1}}\bo{\eta}^{\b{2}}\bo{\eta}^{\b{3}}, \\
  \rho(\alpha\w\beta)&=(-1)^{1+q+pq}s
  \sum_{I,J}^{<}
  \e_{I,I^{\ast}}\e_{J,J^{\ast}}
   \alpha^I_{\b{J}}\, \beta^{I^{\ast}}_{\b{J}^{\ast}}
\bo{\eta}^{\b{1}}\bo{\eta}^{\b{2}}\bo{\eta}^{\b{3}}.
 \end{align*}
 \end{proof}

An immediate corollary is the useful formula 
\begin{equation}
 \Tr(\alpha\w\beta)=(-1)^{q+1}\int_X\rho(\alpha)\w\rho(\beta), \quad
  \alpha\in B^{p,q}, \ \beta\in B^{3-p,3-q},
\end{equation}
from which we obtain the  ``integration by parts'' formulas below:
  \begin{align}
 \Tr(\b{\pa}\alpha\w \beta)&=(-1)^{|\alpha|+1}\Tr(\alpha\w\b{\pa}\beta), \label{pidelbar}\\
 \Tr(\Delta\alpha\w\beta)&=(-1)^{|\alpha|}\Tr(\alpha\w\Delta\beta), \label{pidelta}\\
 \Tr(S\alpha\w\beta)&=(-1)^{|\alpha|+1}\Tr(\alpha\w S\beta), \label{pis}\\
 \Tr(R\alpha\w\beta)&=(-1)^{|\alpha|}\Tr(\alpha\w R\beta). \label{pir}
\end{align}

Now we are in a position  to write down the action of the KS gravity theory \cite{BCOV}.
For simplicity, we here restrict ourselves to the classical theory, which means that the ``string field''
$\va$ takes values in $B^{1,1}$.
It is known  that in closed string field theory, to obtain a gauge invariant action,
we must put the subsidiary condition
$\Delta \va=0$ on the state $\va\in \B$,  
so that the configuration space must be reduced to
\begin{equation}
 B^{1,1}\cap \Ker \Delta=\rho^{-1}(\H^{2,1})\oplus \Im \Delta \cap B^{1,1},
  \label{subsidiary}
\end{equation}
where the first direct summand  of the right hand side is called massless, while the second  massive \cite{BCOV}.
According to (\ref{subsidiary}), we decompose $\va\in B^{1,1}\cap \Ker\Delta$ as
$\va=\va_1\oplus \varPhi$, where $\va_1\in \rho^{-1}(\H^{2,1})$ and $\varPhi\in \Im\Delta$.
Then the action evaluated at $\va$ is given by
\begin{align}
 S_X[\va_1|\varPhi]=\Tr\left(\frac16 (\va_1+\varPhi)\w(\va_1+\varPhi)\w(\va_1+\varPhi)
 -\frac12\b{\pa}\Delta^{-1}\varPhi\w \varPhi\right).
 \label{ksaction}
\end{align}
We note that $\b{\pa}$ above is the BRST operator of the string theory.
From the Hodge decomposition (\ref{hodgedecomposition}) and the partial integration (\ref{pidelbar}),
we can see that the value of  (\ref{ksaction}) does not depend on the choice of $\Delta^{-1}\varPhi$. 
Note that  from (\ref{idA}) we can write $\va_1=\hat{\Pi}(\va)$ and
 $\varPhi=\Delta S\hat{\G}(\va)$.

 We note  that in (\ref{ksaction})  the massless mode $\va_1$ does not has a kinetic term
 so that it acts as a background field \cite{BCOV}.
 Let $\mathcal{M}$ be the moduli space of the complex structures of Calabi-Yau manifolds, and $T\mathcal{M}$
 be its holomorphic tangent bundle.
 Then the complex structure of $X$ determines a point $[X]\in \mathcal{M}$ and 
$\va_1$ an element of $T_{[X]}\mathcal{M}$.
Thus the total background of the action (\ref{ksaction})
 can be regarded as the total space of $T\mathcal{M}$ \cite{BeSa}.
 It is also interesting to note that integration of the holomorphic anomaly equation
 is performed by quantization of the massless modes \cite{BCOV,YY,ASYZ}.

 The equation of motion, which is known as the  Kodaira-Spencer equation,
 construction of its solutions,  and deformations of the action (\ref{ksaction}),
 we will discuss  in later sections.

 \section{Algebraic Structure of  $\bo{\B}$}
 \subsection{Linear operators in sigma model variables}
 We have seen in the previous section that the space of differential forms $\A$
 have  a rich set of  linear operators on it,
 which we can transfer to those on $\B$
 via the isomorphism $\rho:\B\to \A$ (\ref{isomone}--\ref{isomfour}).
 In this subsection, we describe these linear operators on $\B$ as those acting on
 the bosonic variables $z^i,z^{\b{j}}$ as well as the fermionic ones $\bo{\theta}_i, \bo{\eta}^{\b{j}}$.
 
 \subsubsection{The Hodge dual operator}
 The action of the Hodge dual operator 
$\hat{\ast}: B^{p,q}\to B^{q,p}$ 
clearly corresponds to the exchange of  $\bo{\theta}$ and $\bo{\eta}$.
 Form this observation, we find  the following description.
 
 First define a map  $\kappa$ between the sigma model variables by
 \begin{equation}
  \kappa(\bo{\theta}_i)=-\sum g_{i\b{j}}\, \bo{\eta}^{\b{j}}, \quad
  \kappa(\bo{\eta}^{\b{j}})=\sum g^{i\b{j}}\bo{\theta}_i. 
 \end{equation}
Then we extend $\kappa$ to $\B$ as a super-algebra homomorphism.
Finally the action of the transferred Hodge dual operator on a homogeneous element $\alpha\in \B$ is given by
\begin{equation}
 \hat{\ast}\alpha=\i e(|\alpha|)\kappa(\alpha),
  \label{hodgedual}
\end{equation}
where $e(n)$ is the ``phenomenological'' sign factor defined by
$e(n)=(-1)^{(n+1)(n+2)/2}$.
We can easily see that \eqref{hodgedual} satisfies 
 $\hat{\ast}^2\alpha=(-1)^{|\alpha|+1}\alpha$.

The operator $\hat{\ast}$ is well-defined for the wedge product of $\B$.
Indeed from the formula $e(n)e(m)e(n+m)=(-1)^{nm+1}$,  we have
\begin{equation}
 \hat{\ast}(\alpha\w\beta)=\i\, (-1)^{|\alpha|\cdot|\beta|}\,\hat{\ast}(\alpha)\w \hat{\ast}(\beta).
\end{equation}

 \subsubsection{Four differential operators}
 \label{fourdiffop}
 The exterior differentials and their Hermite conjugates viewed in $\B$ are given by
\begin{align}
 \b{\pa}&=\sum \bo{\eta}^{\b{i}}\frac{\pa}{\pa z^{\b{i}}},
 \label{fourone}\\
 \Delta&= s^{-1}\circ \sum\frac{\pa}{\pa z^i}\frac{\pa}{\pa \bo{\theta}_i}\circ s
 \label{fourtwo}\\
 S&=-\sum g^{\b{j}k}\bo{\theta}_k\left(\frac{\pa}{\pa z^{\b{j}}}
 -\Gamma^{\b{i}}_{\b{j}\b{l}}\,\bo{\eta}^{\b{l}}\frac{\pa}{\pa \bo{\eta}^{\b{i}}}\right),
 \label{fourthree}\\
 R&=e^{-\sigma}\circ \sum g^{\b{i}j}\left(\frac{\pa}{\pa z^j}
 +\Gamma^k_{jl}\,\bo{\theta}_k\frac{\pa}{\pa \bo{\theta}_l}\right)
 \frac{\pa}{\pa \bo{\eta}^{\b{i}}} \circ e^{\sigma},
 \label{fourfour}
\end{align}
where 
$\Gamma^i_{jk}=\sum_{\b{l}} g^{i\b{l}}\pa_jg_{k\b{l}}$, $\pa_j=\pa/\pa z^j$,
 is the Christoffel symbol, and it is  unavoidable to use sometimes Einstein's summation rule.

 \subsubsection{ Lefshetz operators}
 The Lefshetz operators,  $\hat{L}:B^{p,q}\to B^{p-1,q+1}$,
 $\hat{\varLambda}:B^{p,q}\to B^{p+1,q-1}$,  are given by the following differential operators
 \begin{align}
  \hat{L}&=\i \sum g_{i\b{j}}\,\bo{\eta}^{\b{j}}\frac{\pa}{\pa \bo{\theta}_i},
  \label{el}\\
  \hat{\varLambda}&=-\i \sum g^{i\b{j}}\bo{\theta}_i\frac{\pa}{\pa \bo{\eta}^{\b{j}}}.
  \label{lambda}
 \end{align}
 Let us check the first formula \eqref{el}.
Pick up an element of $B^{p,q}$ and differentiate it by $\bo{\theta}_i$
\begin{equation*}
 \alpha=\sum_{I,J}^{<}\alpha^I_{\b{J}}\,\bo{\theta}_I\bo{\eta}^{\b{J}}\in B^{p,q},\quad
\frac{\pa}{\pa\bo{\theta}_i} \alpha=\sum_{I,J:i\in I}^{<}\alpha^I_{\b{J}}\,\e(I_i)\bo{\theta}_{I_i}\bo{\eta}^{\b{J}},
\end{equation*}
 where $I_i=I\backslash\{i\}$, and $\e(I_i)$ is a sign factor.
 Then we have
 \begin{equation}
  \sum_{i,j}\i g_{i\b{j}}\,\bo{\eta}^{\b{j}}\frac{\pa}{\pa\bo{\theta}_i} \alpha
   =(-1)^{p-1}\sum_{i,j}\sum_{I,J:i\in I, j\not\in J}^{<}\i g_{i\b{j}} \,\alpha^I_{\b{J}}\,\e(I_i)\bo{\theta}_{I_i}\bo{\eta}^{\b{j}}\bo{\eta}^{\b{J}}.
   \label{upper}
 \end{equation}
 On the other hand, the action of $\hat{L}$ is computed as
 \begin{align}
  \rho(\alpha)&=(-1)^{p(p-1)/2}\sum_{I,J}^{<} s \alpha^I_{\b{J}}\, \e_{I,I^{\ast}}\bo{\eta}^{I^{\ast}}\bo{\eta}^{\b{J}}, \nonumber\\
   \omega\w\rho(\alpha)&=(-1)^{p(p-1)/2+3-p} \sum_{i,j}\sum_{I,J:i\in I, j\not\in J}^{<}
   \i g_{i\b{j}}\, s \alpha^I_{\b{J}}\, \e_{I,I^{\ast}} \bo{\eta}^i\bo{\eta}^{I^{\ast}}\bo{\eta}^{\b{j}}\bo{\eta}^{\b{J}},\nonumber\\
    \rho^{-1}( \omega\w\rho(\alpha))
   &= \sum_{i,j}\sum_{I,J:i\in I, j\not\in J}^{<}
   \i g_{i\b{j}}\,  \alpha^I_{\b{J}}\, \e_{I,I^{\ast}}
   \e'(I_i)\bo{\theta}_{I_i}
   \bo{\eta}^{\b{j}}\bo{\eta}^{\b{J}},
   \label{lower}
 \end{align}
 where $\e'(I_i)$ is the sign factor which satisfies
 $\rho(\bo{\theta}_{I_i})=(-1)^{(p-1)(p-2)/2} s\, \e'(I_i)\bo{\eta}^i\bo{\eta}^{I^{\ast}}$.

 If $i\in I$ is the $k$th element, then the sign factors are $\e(I_i)=(-1)^{k-1}$,
 $\e'(I_i)=(-1)^{p-k}\e_{I,I^{\ast}}$,
 which shows that \eqref{upper} and \eqref{lower} coincide with each other.

Then the relation $\hat{\ast}\circ \hat{\varLambda}=\hat{L}\circ \hat{\ast}$ yields the second formula \eqref{lambda}. 
 We can also verify the Hodge-K\"{a}hler identities (\ref{hodgekahlerone}--\ref{hodgekahlerfour})
 among the converted operators $\hat{L}, \hat{\varLambda}$, and $\b{\pa}, \Delta, S,R$.

 The configuration space $\B$ of the KS gravity
 is a subspace of  the state space of the topological B sigma model
 with stringy excitations omitted,
so that  $\B$ has some remnants of the $N=2$ supersymmetry algebra,
among which are the left/right $U(1)$ charges
 $(1/2)(\pm \i(\hat{L}-\hat{\varLambda})+P_{\B}+Q_{\B})$ \cite{BCOV},
the sum of which yields the ghost number operator
 \begin{equation}
  P_{\B}+Q_{\B}=\sum\bo{\theta}_i\frac{\pa}{\pa \bo{\theta}_i}
   +\sum\bo{\eta}^{\b{j}}\frac{\pa}{\pa \bo{\eta}^{\b{j}}}.
 \end{equation}

 As $\widehat{Q_{\A}}=Q_{\B}$, $\widehat{P_{\A}}=3-P_{\B}$,
 the generators of the $S\!U(2)$ \eqref{angmom} become
 \begin{equation}
  \hat{J_1}=\frac12(\hat{\varLambda}+\hat{L}), \quad
   \hat{J_2}=\frac{1}{2\i}(\hat{\varLambda}-\hat{L}), \quad
   \hat{J_3}=\frac12(P_{\B}-Q_{\B}).
 \end{equation}
 Notice that these operators are of order one \cite{Ro},
 which means that they act on $\B$ as derivations with respect to the wedge product.
 Therefore it can be said that the wedge product of $\B$ conserves the $S\!U(2)$ symmetry.

\subsubsection{Hodge duality,  $\bo{SU(2)}$, and $\bo{B^{1,1}}$}
Note that the Hodge dual operator $\hat{\ast}$ maps each diagonal sector $B^{p,p}$ to itself
with the eigenvalues $\pm \i$.
It is easy to see that $\hat{\ast}$ acts as $-\i$ on $B^{0,0}$, and as $+\i$ on $B^{3,3}$.
 As for the most important sector $B^{1,1}$, we find the following fact:
\begin{lemma}
\label{Boneone}
 For $\alpha\in B^{1,1}$, $\hat{\ast}\alpha=+\i\alpha$ \ if and only if \
 $\hat{L}\alpha=0$.
\end{lemma}
\begin{proof}
 Let us denote an element of $B^{1,1}$ by $\alpha=\sum\alpha^i_{\b{j}}\bo{\theta}_i\bo{\eta}^{\b{j}}$.
 Then $\hat{L}\alpha=\i \sum g_{i\b{k}}\alpha^i_{\b{j}}\bo{\eta}^{\b{k}}\bo{\eta}^{\b{j}}$, so that
 $\hat{L}\alpha=0$ $\Leftrightarrow$ $\sum g_{i\b{k}}\alpha^i_{\b{j}}=\sum g_{i\b{j}}\alpha^i_{\b{k}}$.
 On the other hand,  from  $\hat{\ast}\alpha=\i \sum g_{i\b{k}}g^{\b{j}l}\alpha^i_{\b{j}}\bo{\theta}_l\bo{\eta}^{\b{k}}$,
 we see that $\hat{\alpha}=+\i\alpha$ $\Leftrightarrow$
 $\sum g_{i\b{k}}g^{\b{j}l}\alpha^i_{\b{j}}=\alpha^l_{\b{k}}$. 
\end{proof}
In other words,  $\hat{\ast}\alpha=+\i\alpha$  if $\alpha$ belongs to a trivial representation (singlet),
and $\hat{\ast}\alpha=-\i\alpha$ if $\alpha$ belongs to an adjoint representation (triplet) of $S\!U(2)$.

\subsection{dGBV algebras on  $\bo{\B}$}
\subsubsection{Algebra associated with the pair  $\bo{(\Delta,\b{\pa})}$}
We recall here the dGBV algebra on $\B$ discovered in \cite{BaKo}.
A good reference of  the subject is \cite{Man}.
Let us define the odd bracket on $\B$ by
\begin{equation}
 [\alpha\bullet\beta]=(-1)^{|\alpha|}\Delta(\alpha\w\beta)-(-1)^{|\alpha|}\Delta(\alpha)\w\beta-\alpha\w\Delta(\beta),
  \label{oddbracket}
\end{equation}
where the symbol  $\bullet$ carries the ghost number minus one.
It should be remarked here that if $\alpha,\beta\in \Ker\Delta$,
then $[\alpha\bullet\beta]=(-1)^{|\alpha|}\Delta(\alpha\w\beta)\in \Im\Delta$.

This bracket satisfies the following relations:
\begin{align}
 [\beta\bullet\alpha]&=-(-1)^{(|\alpha|+1)(|\beta|+1)} \, [\alpha\bullet\beta],\\
 [\alpha\bullet[\beta\bullet\gamma]]&=[[\alpha\bullet\beta]\bullet\gamma]
 +(-1)^{(|\alpha|+1)(|\beta|+1)} \, [\beta\bullet[\alpha\bullet\gamma]],
 \label{jacobi}\\
 [\alpha\bullet(\beta\w\gamma)]&=[\alpha\bullet\beta]\w\gamma
 +(-1)^{(|\alpha|+1)|\beta|}\, \beta\w[\alpha\bullet\gamma],
 \label{seventermone}\\
 \Delta[\alpha\bullet\beta]&=[\Delta\alpha\bullet\beta]+(-1)^{|\alpha|+1} \, [\alpha\bullet\Delta \beta],\\
 \b{\pa}[\alpha\bullet\beta]&=[\b{\pa}\alpha\bullet\beta]+(-1)^{|\alpha|+1} \, [\alpha\bullet\b{\pa} \beta].
 \label{derivation}
\end{align}
The first two relations define a structure of the odd Lie algebra on $\B$,
the third means that $[\alpha\bullet \ ]$ is a derivation with respect to the wedge product,
while the last two show that both $\Delta$ and $\b{\pa}$ are derivations with respect to the odd bracket.

A most useful presentation of the bracket \eqref{oddbracket} is given by the odd Poisson bracket
\begin{equation}
 [\alpha\bullet\beta]=\sum\left(
			   \alpha\frac{\overset{\leftarrow}{\pa}}{\pa z^i}\w \frac{\overset{\rightarrow}{\pa}}{\pa \bo{\theta}_i}\beta
			   -\alpha\frac{\overset{\leftarrow}{\pa}}{\pa \bo{\theta}_i}\w \frac{\overset{\rightarrow}{\pa}}{\pa z^i}\,\beta
				  \right).
 \label{oddpoisson}
\end{equation}
Note that the right hand side does not depend on the holomorphic three-form $\varOmega$.

Recall that $\B$ is a hybrid of polyvectors and differential forms \eqref{b},
and polyvectors  close themselves under the Schouten-Neienhuis bracket \cite{Ro,Fu}.
This puts to $\B$ another product structure; it is essentially the same as 
the odd bracket \eqref{oddbracket}.
To see this, let us use multi-indices to  present  two homogeneous elements of $\B$
 \begin{equation}
  \alpha=\sum_{|I|=p,|J|=q}^{<}\alpha^I_{\b{J}} \,\bo{\theta}_{I} \bo{\eta}^{\b{J}} \in B^{p,q}, \quad
\beta=\sum_{|K|=r,|L|=s}^{<}\beta^K_{\b{L}} \,\bo{\theta}_{K} \bo{\eta}^{\b{L}} \in B^{r,s},
 \end{equation}
 Then the odd bracket can be written as
 \begin{equation}
  [\alpha\bullet\beta]=-(-1)^{q(r+1)}\sum_{I,J,K,L}^{<}
   [\![\alpha^I_{\b{J}} \bo{\theta}_{I}, \beta^K_{\b{L}} \bo{\theta}_{K}]\!]_{\mathrm{sn}}\,
    \bo{\eta}^{\b{J}} \bo{\eta}^{\b{L}},
   \label{tian}
 \end{equation}
where $[\![ \ , \  ]\!]_{\mathrm{sn}}$ is the Schouten-Neienhuis bracket.
The formula \eqref{tian} is known as the generalized Tian's Lemma.
The original Tian's Lemma refers to the case of $p=r=1$ \cite{Ti},
where the Schouten-Neienhuis bracket reduces to the Lie bracket of vector fields.

\subsubsection{Algebra associated with  the pair $\bo{(R,S)}$}
The dGVB algebra in the previous subsubsection has been constructed using the two operators:
$\Delta$ of ghost number minus one and of order two \cite{Ro} and $\b{\pa}$ of ghost number one and of order one,
which satisfy $\Delta^2=0$, $\b{\pa}^2=0$ and $\Delta\b{\pa}+\b{\pa}\Delta=0$.

We have another pair $(R, S)$ which has the same properties as $(\Delta,\b{\pa})$,
which leads us to define the following odd bracket  
\begin{equation}
 [\alpha\dia\beta]=(-1)^{|\alpha|}R(\alpha\w\beta)-(-1)^{|\alpha|}R(\alpha)\w\beta-\alpha\w R(\beta).
  \label{anotheroddbracket}
\end{equation}

The two odd brackets are related by the Hodge dual operator as
\begin{equation}
 [\alpha\dia\beta]=\i\, (-1)^{|\alpha|\cdot|\beta|+1}\, \hat{\ast}[\hat{\ast}\alpha\bullet\hat{\ast}\beta],
  \label{twobracketshodge}
\end{equation}
from which we can show the dGBV algebra relations:
\begin{align}
 [\beta\dia \alpha]&=-(-1)^{(|\alpha|+1)(|\beta|+1)} \, [\alpha\dia\beta],\\
 [\alpha\dia[\beta \dia\gamma]]&=[[\alpha\dia\beta]\dia\gamma]
 +(-1)^{(|\alpha|+1)(|\beta|+1)} \, [\beta\dia[\alpha\dia\gamma]],\\
 [\alpha\dia(\beta\w\gamma)]&=[\alpha\dia\beta]\w\gamma
 +(-1)^{(|\alpha|+1)|\beta|}\, \beta\w[\alpha\dia\gamma],
 \label{seventermtwo}\\
 R[\alpha\dia\beta]&=[R\alpha\dia\beta]+(-1)^{|\alpha|+1} \, [\alpha\dia R \beta],\\
 S[\alpha\dia\beta]&=[S\alpha\dia\beta]+(-1)^{|\alpha|+1} \, [\alpha\dia S \beta].
\end{align}

We have also obtained the following formulas
\begin{align}
 &\Delta[\alpha\dia\beta]-[\Delta\alpha\dia\beta]+(-1)^{|\alpha|}[\alpha\dia \Delta\beta] \nonumber\\
 =-&R[\alpha\bullet\beta]+[R\alpha\bullet\beta]-(-1)^{|\alpha|}[\alpha\bullet R\beta],
 \label{formulalaplacian}\\
 &S[\alpha\bullet\beta]-[S\alpha\bullet\beta]+(-1)^{|\alpha|}[\alpha\bullet S\beta]\nonumber\\
 =-&\b{\pa}[\alpha\dia \beta]+[\b{\pa}\alpha\dia\beta]-(-1)^{|\alpha|}[\alpha\dia\b{\pa}\beta]\nonumber\\
 =(-&1)^{|\alpha|}\left(\hat{\lap}(\alpha\w\beta)-\hat{\lap}\alpha\w\beta-\alpha\w\hat{\lap}\beta\right),
 \label{laplacian}
\end{align}
where we recall that the Laplacian is given by
$\hat{\lap}=-\b{\pa}R-R\b{\pa}=\Delta S+S\Delta$.

The odd Poisson bracket form is given by
  \begin{equation}
   [\alpha\dia\beta]
    =\sum    g^{\b{i}j}\left(
\alpha\left(\frac{\overset{\leftarrow}{\pa}}{\pa z^j}+\frac{\overset{\leftarrow}{\pa}}{\pa \theta_l}\theta_k\Gamma^k_{jl}\right)
 \w \frac{\overset{\rightarrow}{\pa}}{\pa \eta^{\b{i}}}\beta
-\alpha\frac{\overset{\leftarrow}{\pa}}{\pa \eta^{\b{i}}}  \w
 \left(\frac{\overset{\rightarrow}{\pa}}{\pa z^j}+\Gamma^k_{jl}\theta_k\frac{\overset{\rightarrow}{\pa}}{\pa \theta_l}\right)\beta
\right),
\label{oddpoissonR}
  \end{equation}
the right hand side of which is independent of the scalar field $\sigma$.

  The dGBV algebras associated with the space of differential forms $\A$,
  which are expected to be a mirror dual to those of $\B$ when the sigma model instanton effects
  are incorporated,  have been studied in \cite{Mer,CZ,P}.

 \section{KS Action and its Deformation}
 The algebraic tools developed in the previous sector enables a systematic treatment
 of the classical KS gravity action
\begin{equation}
 S_X[\varphi_1|\varPhi]=\Tr\left(\frac16\varphi\w\varphi\w\varphi
 -\frac12\b{\pa}\Delta^{-1}\varPhi\w \varPhi\right),
\label{ksactionagain}
 \end{equation}
 where we recall that 
$\varphi=\varphi_1+\varPhi$ is an element of $B^{1,1}\cap \Ker\Delta$,
the massless part $\varphi_1 =\hat{\Pi}(\varphi)\in \rho^{-1}(\mathbb{H}^{2,1})$
is a background field,
and the massive part $\varPhi\in \Im\Delta\cap B^{1,1}$ carries dynamical degrees of freedom.

We obtain the equation of motion of $\varphi$ from variation of the action \eqref{ksactionagain}
\begin{equation}
 \b{\pa}\varphi+\frac12[\varphi\bullet\varphi]=0.
  \label{kseq}
\end{equation}
In components,  $\va=\sum\va^i_{\b{j}}\bo{\theta}_i\bo{\eta}^{\b{j}}$,  it reads
\begin{equation}
 \pa_{\b{j}}\va^i_{\b{k}}-\pa_{\b{k}}\va^i_{\b{j}}
  +\sum_l \left(\va^l_{\b{j}}\pa_l\va^{i}_{\b{k}}-\va^l_{\b{k}}\pa_l\va^{i}_{\b{j}}\right)=0. 
\end{equation}
This is the famous Kodaira-Spencer equation \cite{Kod},
which describes deformation of complex structures on $X$.

The action has the gauge symmetry; the infinitesimal form of it is
\begin{equation}
 \delta \varPhi=\b{\pa}\xi+[\va\bullet \xi],
\end{equation}
where   $\xi\in B^{1,0}\cap \Ker\Delta$ is a gauge parameter.

\subsection{Solution to KS equation}
We construct a solution to the KS equation \eqref{kseq} following \cite{Ti,To}.

Let $\va_1\in \rho^{-1}(\mathbb{H}^{2,1})$ a massless mode. Then we claim that there is a series of massive modes
$\va_n\in \rho^{-1}(\Im(\pa\b{\pa}^{\dag}))\cap B^{1,1}$, $n\geq 2$, such that
 $\va=\sum_{n\geq 1}\va_n$ solves the KS equation \eqref{kseq}.
 In fact, we can solve order by order the equation for $\va_n$:
 \begin{equation}
  \b{\pa}\rho(\va_n)=\psi_n:=\frac12\,\rho\left(\sum_{i+j=n}[\va_i\bullet\va_j]\right).
   \label{pert}
 \end{equation}
For  $n=1$, \eqref{pert} is trivially satisfied by $\va_1$.  

For $n=2$,  \eqref{pert} becomes
 \begin{equation}
  \b{\pa}\rho(\va_2)=\psi_2=\rho([\va_1\bullet\va_1]).
   \label{perttwo}
 \end{equation}
From the definition of the odd bracket \eqref{oddbracket} and the fact $\va_1\in \Ker\Delta$,
we see that  $\psi_2\in \Im\pa$.
The formula \eqref{derivation} shows that   $\psi_2\in \Ker \b{\pa}$
because $\va_1\in \Ker \b{\pa}$.
Then $\pa\b{\pa}$-Lemma implies that $\psi_2\in \Im(\pa\b{\pa})\subset \Im\b{\pa}$,
from which we see that \eqref{perttwo} has solutions.
In particular the gauge condition $\rho(\va_2)\in \Im\b{\pa}^{\dag}$ picks up the unique one:
 $\rho(\va_2)=\b{\pa}^{\dag}\G\psi_2 \in \Im(\pa\b{\pa^{\dag}})$.

Suppose that we have solutions $\va_k=\b{\pa}^{\dag}\G\psi_n\in \rho^{-1}\Im(\pa\b{\pa}^{\dag})$
for $2\leq k\leq n$.
Then we see that $\psi_{n+1}\in \Im\pa$ because each term $[\va_i\bullet\va_j]$
of $\psi_{n+1}$ belongs to $\Im\Delta$ by the inductive assumption.
We can also show that $\psi_{n+1}\in \Ker \b{\pa}$ since
\begin{equation}
 \b{\pa}\,\frac12\sum_{i+j=n+1}[\va_i\bullet\va_j]=\sum_{i+j=n+1}[\b{\pa}\va_i\bullet\va_j]
  =-\frac12\sum_{i+j=n+1}\sum_{k+l=i}[[\va_k\bullet\va_l]\bullet\va_j],
\end{equation}
 the right hand side of which vanishes due to the Jacobi identity of the odd bracket \eqref{jacobi}.
 Therefore again from the $\pa\b{\pa}$-Lemma, we see that $\psi_{n+1}\in \Im(\pa\b{\pa})\subset \Im(\b{\pa})$,
 and we obtain the $n+1$st solution
 $\rho(\va_{n+1})=\b{\pa}^{\dag}\G\psi_{n+1}$.
 Thus we have obtained a solution $\va=\sum_{n\geq 1}\va_n$ of the KS equation \eqref{kseq}.
Note that  $\va$ satisfies $\Delta\va=0$, $R\va=0$ by construction.
 For the proof of the convergence of the infinite sum $\sum_{n\geq 1}\va_n$ for sufficiently small $\va_1$s,
see \cite{It}.

Let us define the massive propagator \cite{BeSa} by
 \begin{equation}
  \P=\rho^{-1}\circ \b{\pa}^{\dag}\pa\G\circ \rho:B^{p,q}\to B^{p-1,q-1}.
   \label{massivepropagator}
 \end{equation}
 Then we get the recursion formula for $n\geq 2$:
 \begin{equation}
  \va_n=\frac12\sum_{i=1}^{n-1}\P(\va_i\w\va_{n-i}).
   \label{recursion}
 \end{equation}

\begin{proposition}
\label{propo}
 The solution $\va=\sum_{n\geq 1}\va_n$ of \eqref{kseq} constructed above satisfies 
 $\hat{\ast}\va=+\i\va$.
\end{proposition}
\begin{proof}
 From the lemma \ref{Boneone},  it suffices to show that  $\hat{L}\va=0$.
 First note that $\hat{L}\va_1=0$ because $\hat{L}$ maps $\rho^{-1}(\H^{2,1})$ to $\rho^{-1}(\H^{3,2})=\{0\}$.
 Second, we observe that
 $\hat{L}$ commutes with $\P$ as $[L,\G]=0$ and $[L,\b{\pa}^{\dag}\pa]=0$,
 so that $\hat{L}\va_2=\frac12\P\hat{L}(\va_1\w\va_1)=\frac12\P(\hat{L}\va_1\w\va_1+\va_1\w\hat{L}\va_1)=0$,
 where we have used the fact that $\hat{L}$ acts as a derivation with respect to the wedge product.
 Finally,  from the recursion relation \eqref{recursion},  we see $\hat{L}\va_n=0$ for each $n$.  
\end{proof}

\subsection{Condensation of string and deformation of  supercharges}
Let $\varPhi[\va_1]=\sum_{n\geq 2}\va_n$, where $\va_1\in \rho^{-1}(\H^{2,1})$ and $\va_n$, $n\geq 2$,
 is the solution of  \eqref{pert} constructed above.
 If we let $\varPhi$ have vacuum expectation value $\varPhi[\va_1]$ and expand the KS action \eqref{ksactionagain}
 with respect to the fluctuation $\varPhi'$ around it, we get \cite{BCOV}
 \begin{align}
  S_X\left[\va_1|\varPhi[\va_1]+\varPhi'\right]
  & =S_X\left[\va_1|\varPhi[\va_1]\right]+\tilde{S}[\varPhi'],
  \label{vacuum}\\
  \tilde{S}[\varPhi']&=\Tr\left(\frac16\varPhi'\w\varPhi'\w\varPhi'+\frac12\Delta^{-1}(\b{\pa}\varPhi'+[\va\bullet \varPhi'])\w\varPhi'\right),
  \label{newvacuum}
 \end{align}
 where $\va=\va_1+\varPhi[\va_1]$ is the solution of the KS equation considered above.
 
 The first term of \eqref{vacuum} is known as the prepotential of topological string and has a following expansion \cite{LoSh}:
 \begin{equation}
  \begin{split}
    S_X\left[\va_1|\varPhi[\va_1]\right]
  &=\frac16\Tr(\va_1\w\va_1\w\va_1)+\frac18\Tr(\va_1\w\va_1\w\P(\va_1\w\va_1))\\
   &+\frac18\Tr(\va_1\w\va_1\w\P(\va_1\w \P(\va_1\w\va_1)))
   +\cdots.
\end{split}
\end{equation} 

We note that in the kinetic term of  $\tilde{S}[\varPhi']$ \eqref{newvacuum}  appearance of the new BRST operator 
$\b{\pa}+[\va\bullet\ ]$,
 which is nilpotent if $\va$ solves the KS equation,
 as is always the case for string field theories \cite{Sen1}.
The same deformation of the BRST operator $\b{\pa}$ has also been 
considered in the first quantization approach to topological B sigma model \cite{LaMa}.  

 In the reference \cite{MY},
 the authors have found that  in addition to the BRST operator $\b{\pa}$, 
 $S$ also deforms to $S+[\va\dia\ ]$ in our notation.
Here we remark  the identities $[\va\bullet \ ]=[\Delta,\va\w]$, $[\va\dia \ ]=[R,\va\w]$.
They have also given the equation for this  to be nilpotent:
 \begin{equation}
  \left(D^i\va^j_{\b{k}}+\va^i_{\b{l}}D^{\b{l}}\va^j_{\b{k}}\right)
   \bo{\theta}_i\bo{\theta}_j\bo{\eta}^{\b{k}}=0,
\label{mercus}
 \end{equation}
 where $D^i=\sum g^{i\b{j}}D_{\b{j}}$ is the covariant derivative of the Levi-Civita connection.
 However we can easily show  that the left hand side of \eqref{mercus} can be written as
 \begin{equation}
  %
-\left(  S\va+\frac12[\va\dia\va]\right)
  =-\i\hat{\ast}\left(\b{\pa}\va+\frac12[\va\bullet\va]\right),
 \end{equation}
 which  shows that \eqref{mercus} does not put any constraint on $\va$ other than the KS equation.
 The remaining operators of order two $\Delta$ and $R$  are left unchanged under the deformation \cite{MY}.
We note that from the relation between two odd  brackets \eqref{twobracketshodge},
it is easily seen that
\begin{equation}
S+[\va\dia \ ]=\hat{\ast}\circ(\b{\pa}+[\va\bullet \ ])\circ \hat{\ast}.
\end{equation}
The commutation relations between the Lefshetz operators and
 $[\va\bullet \ ]$, $[\va\dia \ ]$ are given by 
\begin{alignat}{2}
 [\hat{L},[\va\bullet \ ]]&=[\hat{L}\va\bullet\ ]=0, &\quad
 [\hat{L},[\va\dia \ ]]&=-\i[\va\bullet \ ],\\
 [\hat{\varLambda},[\va\dia \ ]]&=[\hat{\varLambda}\va\dia\ ]=0, &\quad
  [\hat{\varLambda},[\va\bullet \ ]]&=\i[\va\dia \ ],
\end{alignat}
which shows that the Hodge-K\"{a}hler identities  (\ref{hodgekahlerone}--\ref{hodgekahlerfour})
are preserved under the deformation
$\b{\pa}\to \b{\pa}+[\va\bullet \ ]$,
$S\to S+[\va\dia \ ]$.

 We want to  find a ``Calabi-Yau manifold'' $\tilde{X}$ which gives the identity
 $\tilde{S}[\varPhi']=S_{\tilde{X}}[0|\varPhi']$;
 it is clear that $\tilde{X}$ has a close relation to the deformation of the complex structure of $X$
 induced  by the solution of the KS equation $\va$.
Therefore we describe the deformation of complex structure below.

 \section{Deformation of Complex Structure}
\subsection{Classical deformation}
 Let $\va_1\in \rho^{-1}(\H^{2,1})$ be a massless mode of the KS action, so small in magnitude that 
 the solution of the KS equation constructed above $\va=\sum_{n\geq 1}\va_n$ converges,
 and $X_{\va_1}$ the Calabi-Yau manifold with the complex structure defined by $\va$,
 that is, a local function $f$ on $X_{\va_1}$ is holomorphic if
 $(\pa_{\b{i}}+\sum\va^j_{\b{i}}\pa_{j})f=0$ for all $\b{i}$.

 Local frames of $(1,0)$ and $(0,1)$ vector fields and 1-forms on $X_{\va_1}$ are given by
 \begin{alignat}{2}
  \bo{e}_{i}&=\frac{\pa}{\pa z^{i}}+\sum\b{\va}^{\b{j}}_i\frac{\pa}{\pa z^{\b{j}}},
  &\quad
    \bo{e}_{\b{j}}&=\frac{\pa}{\pa z^{\b{j}} }+\sum\va^i_{\b{j}}\frac{\pa}{\pa z^i},\\
  \bo{f}^i&=\d z^i-\sum\va^i_{\b{j}}\d z^{\b{j}},
&\quad
  \bo{f}^{\b{j}}&=\d z^{\b{j}}-\sum\b{\va}^{\b{j}}_i\d z^{i},
 \end{alignat}
 where $\b{\va}^{\b{j}}_i$ is the complex conjugate of  $\va^j_{\b{i}}$.
 Their pairing is
 $\langle\bo{f}^i,\bo{e}_j\rangle=\delta^i_j-\sum\va^i_{\b{k}}\b{\va}^{\b{k}}_j=(I-\va\b{\va})^i_j$.
 If we define the matrix $N=(I-\va\b{\va})^{-1}$, then
 we have the identities
 $\sum N^i_j\va^j_{\b{k}} =\sum \va^i_{\b{l}}\b{N}^{\b{l}}_{\b{k}}$,
 and $\sum g_{m\b{l}}N^m_k=\sum g_{k\b{j}}\b{N}^{\b{j}}_{\b{l}}$,
 where $\b{N}$ is the complex conjugate of $N$,
and 
$\pa/\pa z^i=\sum N^j_i(\bo{e}_j-\b{\va}^{\b{k}}_j\bo{e}_{\b{k}})$, \
$\d z^l=\sum N^l_i(\bo{f}^i+\va^i_{\b{k}}\bo{f}^{\b{k}})$.

We claim that a holomorphic three-form on $X_{\va_1}$ is given by
\begin{equation}
 \varOmega_{\va_1}=s\bo{f}^1 \bo{f}^2 \bo{f}^3
  =\rho\left(1-\va+\frac{1}{2!}\va\w\va-\frac{1}{3!}\va\w\va\w\va\right).
\end{equation}
As $\varOmega_{\va_1}$ is a $(3,0)$-form on $X_{\va_1}$,
we have only to show that $\d\varOmega_{\va_1}=0$,
which is equivalent to
\begin{equation}
 (-\b{\pa}+\Delta) \left(1-\va+\frac{1}{2!}\va\w\va-\frac{1}{3!}\va\w\va\w\va\right)=0.
  \label{hol}
\end{equation}
It is easy to check \eqref{hol} by the KS equation and the dGBV algebra relations.
We also note the volume form that it defines is
 $\i\varOmega_{\va_1}\w\overline{\varOmega}_{\va_1}=\det(I-\va\b{\va})\i \varOmega\w\overline{\varOmega}$.

 The K\"{a}hler form $\omega$ is still of the type $(1,1)$ under the new complex structure,
 and given by
 $\omega=\i\sum g_{k\b{j}}\b{N}^{\b{j}}_{\b{l}} \bo{f}^k \bo{f}^{\b{l}}$.

 The configuration space $\B_{\va_1}$ for the KS gravity defined on $X_{\va_1}$ can be written as
 \begin{equation}
  \B_{\va_1}=\bigoplus_{p,q} B^{p,q}_{\va_1}, \quad
  B^{p,q}_{\va_1}=\left\{\beta= \frac{1}{p!q!}\sum \beta^{i_1,\dots, i_p}_{\b{j}_1,\dots,\b{j}_q}(z,\b{z})\,
  \bo{e}_{i_1}\cdots\bo{e}_{i_p}
  \bo{f}^{\b{j}_1}\cdots \bo{f}^{\b{j}_q}\right\}.
  \label{deformedB}
 \end{equation}
 Note that $\{\B_{\va_1}\}_{\va_1\in \rho^{-1}(\H^{2,1})}$ defines a vector bundle, or a bundle
 of super-algebras on a neighborhood of $[X]$
 in the moduli space $\mathcal{M}$.
It would be nice to know interesting connections on this bundle \cite{RSZ}.

It is also convenient to make a choice of  local holomorphic coordinates $w^{\alpha}$ of $X_{\va_1}$
\cite{BGI,BaLa}.
If we put the relations between local frames  of $T^{\ast}X_{\va_1}$ and  $TX_{\va_1}$ by
\begin{align}
 \bo{f}^{i}&=\sum_{\alpha} A^i_{\alpha}\d w^{\alpha},\\
 \bo{e}_i&=\sum_{\alpha} B_i^{\alpha}\frac{\pa}{\pa w^{\alpha}},
\end{align}
then we have the matrix relation $AB=N^{-1}$ and
 \begin{align}
  \d w^{\alpha}&=(A^{-1})^{\alpha}_i\d z^i-(A^{-1}\va)^{\alpha}_{\b{j}}\d z^{\b{j}},
  \label{relationone}\\
  \d z^{i}&=(B^{-1})^i_{\alpha}\d w^{\alpha}+(\va \b{B}^{-1})^i_{\b{\beta}}\d w^{\b{\beta}}.
  \label{relationtwo}
  \end{align}
Among the consistency conditions derived from \eqref{relationone} are
  \begin{align}
   \pa_i(A^{-1})^{\alpha}_j - \pa_j(A^{-1})^{\alpha}_i&=0,\\
\sum_{\alpha} (A^{-1})^{\alpha}_le_{\b{j}}A^i_{\alpha}&=\pa_l\va^i_{\b{j}},
  \end{align}
  where $e_{\b{j}}=\pa_{\b{j}}+\sum \va^s_{\b{j}}\pa_s$ 
is the differential operator corresponding to the vector field $\bo{e}_{\b{j}}$.
Similarly, from \eqref{relationtwo}, we obtain the consistency condition
 \begin{equation}
  \sum B^{\alpha}_ie_{\b{j}}(B^{-1})^l_{\alpha}=\sum N^l_kM_{i\b{j}}^k,
   \quad M_{i\b{j}}^k:=e_i\va^k_{\b{j}}+\sum \va^k_{\b{m}}e_{\b{j}}\,\b{\va}^{\b{m}}_i.
 \end{equation}

The K\"{a}hler two-form $\omega$ is rewritten as
$\omega=\i \tilde{g}_{\alpha\b{\beta}}\d w^{\alpha}\w\d w^{\b{\beta}}$, where
\begin{equation}
 \tilde{g}_{\alpha\b{\beta}}=\sum_{j,k} A^k_{\alpha}(\b{B}^{-1})^{\b{j}}_{\b{\beta}}\, g_{k\b{j}},\quad
 \tilde{g}^{\alpha\b{\beta}}=\sum_{l,m} \b{B}^{\b{\beta}}_{\b{l}}(A^{-1})^{\alpha}_m \, g^{\b{l}m},
\end{equation}
from which we can compute the Christoffel symbols $\tilde{\Gamma}^{\alpha}_{\beta\gamma}$
for the K\"{a}hler metric $\tilde{g}_{\alpha\b{\beta}}$.
\subsection{Differential operators  in new complex structure}
In this subsection, we compute the four differential operators introduced  in \ref{fourdiffop}
for the KS gravity theory on the deformed Calabi-Yau manifold $X_{\va_1}$, which we denote by
$\b{\pa}_{\va_1}$, $\Delta_{\va_1}$, $S_{\va_1}$, and $R_{\va_1}$,
and show that these are just reduced  to
$\b{\pa}+[\va\bullet\ ]$, $\Delta$, $S+[\va\dia\ ]$, and $R$, respectively
in the holomorphic limit where
we make an analytic continuation of the moduli parameters so that
$\b{\va}_1$ is set to zero while $\va_1$ is kept fixed \cite{BCOV}.

Let us explain  the ``analytic continuation'' above more detail.
Let $\overline{\M}$ be the same manifold as $\M$ with the opposite complex structure.
Then we extend the moduli space to be $\M\times \overline{\M}$ \cite{BCOV,Wi3,Kon}
and the original moduli space is diagonally embedded  as $\M_{\text{diag}}$.
%
Let $(t^{a})$ be local holomorphic coordinates of $\mathcal{M}$,
and the point  $[X]\in \M$ correspond to  $t=t_0$.
By the classical deformation described in the previous subsection, 
$[X]\in \M\times \overline{\M}$, with coordinates $(t_0,\b{t}_0)$, 
 moves along the diagonal line to the point corresponding to $X_{\va_1}$;
 then  setting $\b{\va}_1=0$ we arrive at a point on the horizontal line $\b{t}=\b{t}_0$.
 We call the point $X_{\va_1}^{\text{hol}}$ and the deformation from $X$ to $X_{\va_1}^{\text{hol}}$
 the holomorphic deformation.


It is easy to give the differential operators on $X_{\va_1}$  
in terms of the local holomorphic coordinates $w^{\alpha}$
according to the formulas (\ref{fourone}--\ref{fourfour}).
To perform the holomorphic limit above, however, it is necessary to change
coordinates from $(w^{\alpha},w^{\b{\beta}})$ to $(z^i,z^{\b{j}})$.

The four differential operators acting on $\B_{\va_1}$ of the form \eqref{deformedB}  
are given by
\begin{align}
 \b{\pa}_{\va_1}&=\b{N}^{\b{k}}_{\b{r}}\bo{f}^{\b{r}}
 \left(\frac{\pa}{\pa z^{\b{k}}}+\va^{l}_{\b{k}}\frac{\pa}{\pa z^l}
 -N^p_sM^s_{i\b{k}}\bo{e}_p\frac{\pa}{\pa \bo{e}_i}\right),\\
 \Delta_{\va_1}&=(s\det N^{-1})^{-1}\circ\left(\frac{\pa}{\pa z^i}+\b{\va}^{\b{j}}_i\frac{\pa}{\pa z^{\b{j}}}\right)
 \frac{\pa}{\pa \bo{e}_i} \circ (s\det N^{-1})
 -\bo{f}^{\b{m}}\pa_{\b{m}}\b{\va}^{\b{j}}_i
 \frac{\pa}{\pa \bo{e}_i} \frac{\pa}{\pa \bo{f}^{\b{j}}},\\
 S_{\va_1}&=g^{\b{u}l} \b{N}^{\b{k}}_{\b{u}} \bo{e}_l
 \left(
 -\frac{\pa}{\pa z^{\b{k}}}-\va^{i}_{\b{k}}\frac{\pa}{\pa z^i}
 +N^n_rM_{i\b{k}}^r\bo{e}_n\frac{\pa}{\pa \bo{e}_i}
 +g_{\b{m}n}g^{\b{j}q}\b{N}^{\b{m}}_{\b{s}}M^n_{q\b{k}}\,\bo{f}^{\b{s}}\frac{\pa}{\pa \bo{f}^{\b{j}}}\right)\nonumber\\
 &+ g^{\b{u}l} \b{N}^{\b{k}}_{\b{u}} \bo{e}_l
 \left(\Gamma^{\b{j}}_{\b{k}\b{s}}+g^{m\b{j}}g_{\b{s}q}\Gamma^q_{mn}\va^n_{\b{k}}\right)\bo{f}^{\b{s}}\frac{\pa}{\pa \bo{f}^{\b{j}}}, \\
R_{\va_1}&=e^{-\tilde{\sigma}}\circ g^{k\b{j}} \left(\frac{\pa}{\pa z^k}+\b{\va}^{\b{l}}_k\frac{\pa}{\pa z^{\b{l}}}\right)
 \frac{\pa}{\pa \bo{f}^{\b{j}}}\circ  e^{\tilde{\sigma}}
 +g^{\b{t}k}\pa_{\b{t}}\b{\va}^{\b{j}}_k \frac{\pa}{\pa \bo{f}^{\b{j}}}
 -g^{\b{j}k}\pa_{\b{n}}\b{\va}^{\b{l}}_k\,\bo{f}^{\b{n}}\frac{\pa}{\pa \bo{f}^{\b{j}}}\frac{\pa}{\pa \bo{f}^{\b{l}}} \nonumber \\
 &+N_p^{m}\left(
 g^{\b{j}k}  e_{i}(N^{-1})^p_k
 +g^{\b{l}p}  \b{M}^{\b{j}}_{\b{l}i}
 +g^{\b{j}s}(N^{-1})^k_s   (\Gamma^p_{ik}+g^{\b{l}p}g_{k\b{q}}\Gamma^{\b{q}}_{\b{n}\b{l}}\b{\va}^{\b{n}}_{i})
 \right)
 \bo{e}_{m}\frac{\pa}{\pa \bo{e}_{i}}\frac{\pa}{\pa \bo{f}^{\b{j}}},
\end{align}
where $e^{\tilde{\sigma}}=\det N e^{\sigma}$ and use of the Einstein summation convention is unavoidable.

In  the holomorphic limit $\b{\va}\to 0$,
$\bo{e}_i\to \bo{\theta}_i$, \
$\bo{f}^{\b{j}}\to \bo{\eta}^{\b{j}}$, \
$N\to I$, \
$M_{i\b{j}}^k\to \pa_i\va^k_{\b{j}}$.\\
Let us see the limit of each operators.
The cases of $\b{\pa}$, $\Delta$ and $R$ are easy:
\begin{equation}
 \b{\pa}_{\va_1}\to \bo{\eta}^{\b{k}}\left(\frac{\pa}{\pa z^{\b{k}}}+\va^l_{\b{k}}\frac{\pa}{\pa z^l}
				      -\pa_i\va^p_{\b{k}}\bo{\theta}_p\frac{\pa}{\pa \bo{\theta}_i}\right)
 =\b{\pa}+[\va\bullet \ ],
\end{equation}
where we have used the formula \eqref{oddpoisson} and 
\begin{equation}
 \Delta_{\va_1}\to
  s^{-1}\circ \frac{\pa}{\pa z^i}\frac{\pa}{\pa \bo{\theta}_i} \circ s=\Delta,
\end{equation}
\begin{equation}
 R_{\va_1}\to
  e^{-\sigma}\circ g^{k\b{j}}\frac{\pa}{\pa z^k}\frac{\pa}{\pa \bo{\eta}^{\b{j}}} \circ e^{\sigma}
 +g^{k\b{j}}\Gamma^m_{ik}\,\bo{\theta}_m\frac{\pa}{\pa \bo{\theta}_i}\frac{\pa}{\pa \bo{\eta}^{\b{j}}}=R.
\end{equation}

For the case of $S$, we have
\begin{equation}
\begin{split}
 S_{\va_1} \to
 - &g^{\b{k}l}\bo{\theta}_l
 \left( 
 \frac{\pa}{\pa z^{\b{k}}}
 -\Gamma^{\b{j}}_{\b{s}\b{k}}\bo{\eta}^{\b{s}}\frac{\pa}{\pa \bo{\eta}^{\b{j}}}\right) \\
&\qquad\qquad
 -g^{\b{k}l}\bo{\theta}_l
 \left(-\va^t_{\b{k}}  \frac{\pa}{\pa z^t}
 +\pa_i\va^n_{\b{k}}\bo{\theta}_n\frac{\pa}{\pa \bo{\theta}_i}
 +g_{\b{k}n}g^{\b{j}q}D_q\va^n_{\b{s}}\bo{\eta}^{\b{s}}\frac{\pa}{\pa \bo{\eta}^{\b{j}}}\right).
\end{split}
\label{Sdef}
\end{equation}
On the other hand we know from the formula \eqref{oddpoissonR}
\begin{equation}
 [\va\dia\ ]=-g^{\b{a}b}\va^i_{\b{a}}\bo{\theta}_i
  \left(\frac{\pa}{\pa z^b}+\Gamma^c_{bd}\bo{\theta}_c\frac{\pa}{\pa \bo{\theta}_d}\right)
  +g^{\b{a}b}D_b\va^c_{\b{j}}\bo{\theta}_c\bo{\eta}^{\b{j}}\frac{\pa}{\pa \bo{\eta}^{\b{a}}},
  \label{oddpoissonRder}
\end{equation}
and the difference between the second term of the right hand side of \eqref{Sdef} and \eqref{oddpoissonRder} 
$D_i(g^{\b{k}l}\va^n_{\b{k}})\bo{\theta}_l\bo{\theta}_n\pa/\pa \bo{\theta}_i$ vanishes owing  to   Prop.\ref{propo}.
This shows  that the holomorphic limit of $S_{\va_1}$ is $S+[\va\dia \ ]$.

In conclusion, we have identified the the deformation of the differential operators
observed in \eqref{newvacuum}   with the one
caused by the holomorphic limit of the deformation of the complex structure,
so that we can write $\tilde{S}[\varPhi']=S_{X_{\va_1}^{\text{hol}}}[0|\varPhi']$.

At first sight, it may seem  that for the KS action  on $X_{\va_1}^{\mathrm{hol}}$
we must use the deformed trace map 
\begin{equation}
 \Tr_{\va_1}^{\mathrm{hol}}(\alpha)=\int_X\rho_{\va_1}^{\mathrm{hol}}(\alpha)\w\varOmega_{\va_1},
\end{equation}
where $\rho_{\va_1}^{\mathrm{hol}}(\bo{\theta}_I)=(-1)^{|I|(|I|-1)/2}\e_{I,I^{\ast}}s\bo{f}^{I^{\ast}}$.
However 
 $\Tr_{\va_1}^{\mathrm{hol}}(\alpha)=\Tr(\alpha)$ for any $\alpha\in \B$.


\subsection{Deformation of states}
Let us denote the deformed operators by
$(\b{\pa})^{\text{hol}}_{\va_1}=\b{\pa}+[\va\bullet \ ]$,
$S_{\va_1}^{\text{hol}}=S+[\va\dia \ ]$.
Then it is easy to see that
\begin{equation}
(\b{\pa})^{\text{hol}}_{\va_1}\circ \Delta+\Delta\circ (\b{\pa})^{\text{hol}}_{\va_1}=0,\quad 
S_{\va_1}^{\text{hol}}\circ R+R\circ S_{\va_1}^{\text{hol}}=0.
\end{equation}
We can also show
\begin{equation}
 (\b{\pa})^{\text{hol}}_{\va_1}\circ S_{\va_1}^{\text{hol}} +S_{\va_1}^{\text{hol}}\circ (\b{\pa})^{\text{hol}}_{\va_1}=0.
  \label{commute}
\end{equation}
To see this,
we calculate the action of the left hand side of \eqref{commute} on $\alpha\in \B$;
\begin{align}
 (\b{\pa})^{\text{hol}}_{\va_1} S_{\va_1}^{\text{hol}}(\alpha)
 &=\b{\pa}(S\alpha+[\va\dia \alpha])+[\va\bullet(S\alpha+[\va\dia \alpha])], \nonumber\\
 S_{\va_1}^{\text{hol}}(\b{\pa})^{\text{hol}}_{\va_1}(\alpha)
 &=S(\b{\pa}\alpha+[\va\bullet\alpha])+[\va\dia(\b{\pa}\alpha+[\va\bullet\alpha])]. \nonumber
\end{align}
The sum of the two above becomes
 \begin{equation}
\left\{ (\b{\pa})^{\text{hol}}_{\va_1}, S_{\va_1}^{\text{hol}} \right\}(\alpha)
 =[S\va\bullet\alpha]+[\b{\pa}\va\dia \alpha]
 +[\va\dia[\va\bullet\alpha]]+[\va\bullet[\va\dia\alpha]],
 \label{middle}
\end{equation}
where we have used the formula \eqref{laplacian}.

Using the KS equation of motion 
$\b{\pa}\va=-(1/2)\Delta(\va\w\va)$
and its Hodge dual form
$S\va=-(1/2)R(\va\w\va)$,
we can rewrite the first two terms of the right hand side of \eqref{middle} as
\begin{equation*}
 \begin{aligned}
  & -\frac12[R(\va\w\va)\bullet\alpha]-\frac12[\Delta(\va\w\va)\dia\alpha]\\
&  =\frac12\Delta(R(\va\w\va)\w\alpha)+\frac12R(\Delta(\va\w\va)\w\alpha)
  +\frac12R(\va\w\va)\w\Delta\alpha
  +\frac12\Delta(\va\w\va)\w R\alpha,
\end{aligned}
\end{equation*}
while the last two terms of the right hand side of \eqref{middle} are 
\begin{align*}
 [\va\dia[\va\bullet\alpha]]&=R(\va\w\Delta(\va\w\alpha)-\va\w\va\w\Delta\alpha)
 -\va\w R\Delta(\va\w\alpha)+\va\w R(\va\w\Delta\alpha),\\
 [\va\bullet[\va\dia\alpha]]&=\Delta(\va\w R(\va\w\alpha)-\va\w\va\w R\alpha)
 -\va\w \Delta R(\va\w\alpha)+\va\w \Delta(\va\w R\alpha).
\end{align*}
At this point, we have
\begin{align}
& \left\{(\b{\pa})^{\text{hol}}_{\va_1}, S_{\va_1}^{\text{hol}} \right\}(\alpha)\nonumber\\
 =&\Delta\left(\frac12R(\va\w\va)\w\alpha+\va\w R(\va\w \alpha)-\va\w\va\w R\alpha\right)\nonumber\\
 +&R\left(\frac12\Delta(\va\w\va)\w\alpha+\va\w \Delta(\va\w \alpha)-\va\w\va\w \Delta\alpha\right)\nonumber\\
 +&\frac12R(\va\w\va)\w\Delta\alpha+\va\w R(\va\w\Delta\alpha)\nonumber\\
 +&\frac12\Delta(\va\w\va)\w R\alpha+\va\w \Delta(\va\w R\alpha).
 \label{final}
 \end{align}

 By the seven-term relation \cite{Ro} which  holds for  operators of order two  such as $\Delta$ or $R$,
and is merely the expansion of the dGBV relation \eqref{seventermone} or \eqref{seventermtwo}:
\begin{equation}
  \begin{aligned}
   \Delta(\alpha\w\beta\w\gamma)&=\Delta(\alpha\w\beta)\w\gamma+(-1)^{|\alpha|}\alpha\w\Delta(\beta\w\gamma)
   +(-1)^{|\beta|\cdot(|\alpha|+1)}\beta\w\Delta(\alpha\w\gamma)\\
 &  -\Delta(\alpha)\w\beta\w\gamma-(-1)^{|\alpha|}\alpha\w\Delta(\beta)\w\gamma
   -(-1)^{|\alpha|+|\beta|}\alpha\w\beta\w\Delta(\gamma),
   \label{seventerm}
   \end{aligned}
\end{equation}
we have $R(\va\w\va\w\alpha)=R(\va\w\va)\w\alpha+2\va\w R(\va\w\alpha)-\va\w\va\w R(\alpha)$;
thus the first term of the right hand side of \eqref{final} becomes
\begin{equation*}
 \frac12\Delta R(\va\w\va\w\alpha)-\frac12 \Delta(\va\w\va\w R\alpha).
\end{equation*}
Similarly, the second term of the right hand side of \eqref{final} becomes
\begin{equation*}
 \frac12 R\Delta (\va\w\va\w\alpha)-\frac12 R(\va\w\va\w \Delta\alpha),
\end{equation*}
so that their sum is
\begin{equation}
 -\frac12 \Delta(\va\w\va\w R\alpha)-\frac12 R(\va\w\va\w \Delta\alpha).
  \label{firstsecond}
\end{equation}
Then another use of the seven-term relation \eqref{seventerm} shows that  
\eqref{firstsecond} cancels out the third and the fourth terms of \eqref{final}.

We can define the deformed Lapacian by
\begin{equation}
 \hat{\lap}^{\text{hol}}_{\va_1}=-(\b{\pa})^{\text{hol}}_{\va_1}\circ R -R \circ (\b{\pa})^{\text{hol}}_{\va_1}
  =S^{\text{hol}}_{\va_1}\circ \Delta+\Delta\circ S^{\text{hol}}_{\va_1}.
\end{equation}
To see that the two operators above give the same $\hat{\lap}^{\text{hol}}_{\va_1}$,
   we have only to show
\begin{equation}
 -[\va\bullet R\alpha]-R[\va\bullet\alpha]=[\va\dia \Delta \alpha]+\Delta[\va\dia \alpha]
  \label{laplaceequality}
\end{equation}
for each $\alpha\in \B$.
However this  is the special case of the formula \eqref{formulalaplacian}.

Let us define a linear map $f_{\va_1}:\B\to\B$ labeled by an element
 $\va_1\in \rho^{-1}(\H^{2,1})$ as follows.
First for $\alpha\in \B$, set $\alpha_0=\alpha$,
 and define $\alpha_n$,  $n\geq 1$,  recursively by
 \begin{equation}
  \alpha_n=\sum_{k=1}^{n}\P(\va_k\w\alpha_{n-k}),
   \label{alpharecursion}
 \end{equation}
 where we recall that
 $\P$ is the massive propagator defined in \eqref{massivepropagator},
 and $\va_n$ is the solution of
 \eqref{pert} and given by the recursion relation \eqref{recursion}. 
 Then the map $f_{\va_1}$ is simply given by
 \begin{equation}
  f_{\va_1}(\alpha)=\sum_{n=0}^{\infty}\alpha_n.
   \label{infinitesum}
 \end{equation}
Presumably, the infinite sum in \eqref{infinitesum}
will converge for $\va_1$ so small enough that
the infinite sum $\va=\sum_{n\geq 1}\va_n$ converges.

 In fact, $f_{\va_1}$ has  nice properties as a map on $\Ker\Delta$:
 \begin{proposition}
  \label{closed}
 $f_{\va_1}$ maps  \ $\Ker\Delta\cap\Ker \b{\pa}$ \
 to \ $\Ker\Delta\cap\Ker (\b{\pa})^{\mathrm{hol}}_{\va_1}$.
 \end{proposition}
\begin{proof}
 We will show by  mathematical induction  
  \begin{equation}
   \b{\pa}\alpha_n+\sum_{k=1}^{n}[\va_k\bullet \alpha_{n-k}]=0.
\label{eqdeformation}
  \end{equation}
 %
 By the induction hypothesis and the Jacobi identity \eqref{jacobi},
 we have
 \begin{align}
  \b{\pa}\sum_{k=1}^{n}[\va_k\bullet \alpha_{n-k}]
  &=\sum_{k=1}^{n}\left([\b{\pa}\va_k\bullet \alpha_{n-k}]-[\va_k\bullet \b{\pa}\alpha_{n-k}]\right)
  \nonumber\\
  &=-\frac12\sum_{a+b+c=n}[[\va_a\bullet\va_b]\bullet\alpha_{c}]
  +\sum_{a+b+c=n}[\va_a\bullet[\va_b\bullet\alpha_c]]
  \nonumber\\
  &=0. \nonumber
 \end{align}
 It is also easy to see that $[\va_k\bullet\alpha_{n-k}]\in \Im\Delta$ for each $k$ in \eqref{eqdeformation}.
 Then the $\Delta\b{\pa}$-Lemma implies that
\begin{equation}
 \sum_{k=1}^{n}[\va_k\bullet\alpha_{n-k}]\in \Im(\b{\pa}\Delta)\subset \Im(\b{\pa}).
\end{equation} 
 Therefore
 \begin{equation}
  \b{\pa}\b{\pa}^{\dag}\G\rho \sum_{k=1}^{n}[\va_k\bullet\alpha_{n-k}]
   =\rho \sum_{k=1}^{n}[\va_k\bullet\alpha_{n-k}],
 \end{equation}
which is equivalent to \eqref{eqdeformation}.
\end{proof}
\begin{proposition}
$f_{\va_1}$ maps  \ $\Ker\Delta\cap\Im \b{\pa}$ \
 to \ $\Ker\Delta\cap\Im (\b{\pa})^{\mathrm{hol}}_{\va_1}$.
\end{proposition}
 \begin{proof}
  In this case, 
  $\alpha_0\in \Ker\Delta\cap\Im \b{\pa}=\Im(\b{\pa}\Delta)$.
  Thus we can set $\alpha_0=\b{\pa}\beta_0$, $\beta_0\in \Im(\Delta)$.

  We should  find a series $\beta_1,\beta_2,\beta_3,\dots$, which satisfies the equation
  \begin{equation}
   f_{\va_1}(\alpha)=\sum_{n=0}^{\infty}\alpha_n
    =\b{\pa}\sum_{n=0}^{\infty}\beta_n+\sum_{m=1}^{\infty}\sum_{n=0}^{\infty}[\va_m\bullet \beta_n]=0,
    \nonumber
  \end{equation}
  the $n$th order terms of which are
 \begin{equation}
  \alpha_n=\b{\pa}\beta_n+\sum_{k=1}^n[\va_k\bullet\beta_{n-k}].
   \label{exactsol}
 \end{equation}
However, it is easily verified  that $\beta_n$s defined  by the recursion relation
  \begin{equation}
   \beta_n=\P\sum_{k=1}^{n}(\va_k\w\beta_{n-k})
  \end{equation}
solve the equation \eqref{exactsol}.
 \end{proof}
To sum up, the diagram below is commutative:
\begin{equation}
\begin{CD}
\Ker\Delta@> f_{\va_1} >> \Ker\Delta\\
@V\b{\pa} VV           @VV(\b{\pa})^{\mathrm{hol}}_{\va_1} V\\
\Ker \Delta@>f_{\va_1}>> \Ker\Delta.
\end{CD}
\end{equation}
We can also show that $f_{\va_1}\circ \hat{\ast}=\hat{\ast}\circ f_{\va_1}$ from
$\hat{\ast}\circ \P=-\P\circ\hat{\ast}$, and
that
$f_{\va_1}$ maps  \ $\rho^{-1}(\H)=\Ker\hat{\lap}$ \
  to \ $\Ker\Delta\cap \Ker(\b{\pa})^{\mathrm{hol}}_{\va_1}\cap \Ker R\cap\Ker S^{\mathrm{hol}}_{\va_1}$.
 %
%
%

\section{Outlook}
In this paper 
we have analyzed
the equation of motion and deformations of the classical KS gravity theory
using the algebraic structure
of the configuration space $\B$.
Quantization of the KS gravity can be performed by the Batalin-Vilkovisky formalism \cite{Bo,CBTh,Ha,Zw,BCOV,BeSa},
where we relax the condition on the ghost number of the field;
in the action \eqref{ksaction}, the field $\va$ has a contribution from each sector $B^{p,q}\cap \Ker\Delta$,
except for $p\ne 3$, where $\Delta^{-1}$ cannot be  defined \cite{BCOV}.
Then the massless part of the field $\va$ necessarily induces a extended deformation of the Calabi-Yau manifold
\cite{Wi1,Ran,BaKo}.
It should be clear that the algebraic tools developed in this paper
are useful  in the quantization problem of the KS gravity.

It would also  be  interesting to study
the open-closed topological B-model \cite{Hof}
 from the point of view of the second quantization \cite{Wi2,Zw2}.
\vs{0.5cm}

 \noindent
 \textbf{\large Acknowledgements}\\
The author would like to thank  S. Mizoguchi for valuable discussions.



\end{document}